\documentclass{article}
\usepackage{xparse}

\newcommand{\adv}[1]{\mathrm{Adv}\left(#1 \right)}

\newcommand{\herm}[1]{\mathbb{H}^{#1}}

\newcommand{\norm}[1]{\left\Vert #1 \right\Vert}

\newcommand{\psd}[1]{\mathbb{S}^{#1}}

\newcommand{\cvec}[1]{\mathbb{C}^{#1}}

\newcommand{\decomp}[2][{\alpha[d]}]{\sum_{d\in D} #1 \ket{d}\otimes\ket{#2_d}}

\newcommand{\did}{x\in D}

\newcommand{\seqd}[1]{\left( #1 \right)_{\did}}

\NewDocumentCommand{\coll}{ O{0} O{d} O{\ket}}{#3{d}\otimes#3{\psi^{#1}_{#2}}}

\newcommand{\ps}{\sqrt{p_{\mathrm{true}}}}
\newcommand{\kbp}[1][\ps]{\ket{#1}\bra{#1}}

\newcommand{\bigT}[1]{\theta\left(#1\right)}

\newcommand{\bool}{\mathbb{Z}_2}

\usepackage{amsmath}
\usepackage{amssymb}
\usepackage{amsthm}
\usepackage{amsfonts}
\usepackage{physics}
\usepackage{array}
\usepackage{hyperref}

\newtheorem{thm}{Theorem}
\newtheorem{lemma}{Lemma}

\title{One Weird Trick Tightens the Quantum Adversary Bound, Especially for Success Probability Close to $1/2$}
\author{Duyal Yolcu\thanks{\url{https://github.com/qudent}}}

\begin{document}
\maketitle
\begin{abstract}
The textbook adversary bound for function evaluation \cite{ambainis2000quantum, childs2017lecture, negative_weights} states that to evaluate a function $f\colon D\to C$ with success probability $\frac{1}{2}+\delta$ in the quantum query model, one needs at least $\left( 2\delta -\sqrt{1-4\delta^2} \right) \adv{f}$ queries, where $\adv{f}$ is the optimal value of a certain optimization problem. For $\delta \ll 1$, this only allows for a bound of $\bigT{\delta^2 \adv{f}}$ even after a repetition-and-majority-voting argument. In contrast, the polynomial method can sometimes prove a bound that doesn't converge to $0$ as $\delta \to 0$. We improve the $\delta$-dependent prefactor and achieve a bound of $2\delta \adv{f}$. The proof idea is to "turn the output condition into an input condition": From an algorithm that transforms perfectly input-independent initial to imperfectly distinguishable final states, we construct one that transforms imperfectly input-independent initial to perfectly distinguishable final states in the same number of queries by projecting onto the "correct" final subspaces and uncomputing. The resulting $\delta$-dependent condition on initial Gram matrices, compared to the original algorithm's condition on final Gram matrices, allows deriving the tightened prefactor.
\end{abstract}
\section{Outline}
\begin{itemize}
	\item Section \ref{sec:the-problem-function-evaluation} describes the problem of function evaluation in the quantum query model; Table \ref{table:notation_funceval} contains a reference of the symbols defined therein.
	\item Section \ref{sec:the-new-lemma-from-perfect-initial-to-perfect-final-states} proves Lemma \ref{lem:initial_to_final}, capturing the central new idea of this article: Assuming that an algorithm solving such a problem exists, we may modify the initial state it is applied to so that the resulting modified final states allow solving the problem with error probability $0$.
	\item Table \ref{table:notation_linalg} introduces more generic linear algebra notation that is relevant in subsequent sections. Section \ref{sec:a-matrix-norm-inequality} proves Lemma \ref{lem:matrix_norm_inequality}, a technical lemma regarding matrix norms that we'll be able to plug into the adversary bound. The proof uses a result from \cite{bhatia2009positive}.
	\item Section \ref{sec:the-adversary-bound-for-state-conversion} introduces the adversary bound for state conversion.
	\item Section \ref{sec:putting-it-together} puts together the previous results to obtain Theorem \ref{thm:adversary_function}, the central result of this article.
	\item Section \ref{sec:discussion} discusses the result and compares it to the previous and other proofs.
\end{itemize}

\section{The problem: Function evaluation}\label{sec:the-problem-function-evaluation}
\begin{table}[h!]
	\centering
	\begin{tabular}{| m{2.5cm} | m{5cm} |}
		\hline
		$f\colon D\to C$ & Target function, $D\subseteq \bool^n$, $n\in\mathbb{N}$, $|C|<\infty$\\
		\hline
		$\did$ & Input value\\
		\hline\vspace{2pt}
		$\cvec{Q}$ & Query space, $Q:=\{0,1,\ldots n\}$\\\hline
		$O_x \in \cvec{Q\times Q}$ & Unitary oracle for input $\did$ as in Equations \ref{eq:oracle_1}--\ref{eq:oracle_2} \\
		\hline
		$\cvec{W}$ & Ancilla space, $W$ arbitrarily large\\\hline
		\vspace{2pt}
 		$U^0$, \ldots, $U^T\in \cvec{(Q\times W)\times (Q\times W)}$ & Initial, \ldots, final algorithm unitary as in Equation \ref{eq:evolution}\\\hline\vspace{2pt}
		$\ket{\xi}\in\cvec{Q\times W}$ & Initial state\\\hline
		$\ket{\tau_x}\in\cvec{Q\times W}$ & Final state for input $\did$ as in Equation \ref{eq:evolution}\\\hline
		$P_c$ & Projector in final measurement, measurement result $c\in C$, Equation \ref{eq:function_measurements}\\\hline
		$\ket{\ps}\in \cvec{D}$ & Vector whose entries are square roots of success probabilities, i.e. $\braket{x}{\ps} := \norm{P_{f(x)} \ket{\tau_x}}$\\\hline
		$\delta>0$ & Success bias, i.e. lower bound on success probability minus $1/2$ for all $\did$ as in Equation \ref{eq:function_measurements}\\\hline
	\end{tabular}
	\caption{Notation introduced in Section \ref{sec:the-problem-function-evaluation}. Table \ref{table:notation_linalg} contains more notation not used before Section \ref{sec:a-matrix-norm-inequality}.}\label{table:notation_funceval}
\end{table}
We describe the model; Table \ref{table:notation_funceval} contains a reference of symbols introduced in this section.
We consider a query problem to be described by a target function $f\colon D\to C$, where $D\subseteq \bool^n$ is a collection of bit-strings of length $n<\infty$ and $C$ is a finite set. We define a \emph{query space} $\cvec{Q}$ with $Q:=\{0,1,\ldots n\}$; in a problem instance, the input $x\in D$ is encoded in a \textit{phase oracle} unitary $O_x \colon \cvec{Q} \to \cvec{Q}$ defined by
\begin{align}
	O_x\ket{j} &:= (-1)^{x_j} \ket{j} ~~~~~ \mathrm{for~} j<n,\label{eq:oracle_1}\\
	O_x\ket{n} &:= \ket{n}.\label{eq:oracle_2} 
\end{align}
A quantum query algorithm solving the problem acts on the query space together with an \emph{ancilla space} $\cvec{W}$, which is arbitrarily large (in other words, the computer is not restricted regarding the size of the internal memory it uses). It starts with an $\did$-independent\footnote{In the traditional setup; as we'll see, generalizations allow for input-dependent initial states.} initial state $\ket{\xi}\in \cvec{Q\times W}$ and applies a sequence of unitaries $U^0$, $U^1$, $U^2$, \ldots, $U^{T}\colon \cvec{Q\times W}\to \cvec{Q\times W}$ which act alternatingly with the oracle unitary $O_x$ (where $O_x$ acts trivially on the ancilla space $\cvec{W}$). For a given $x\in D$, this leads to a state
\begin{equation}\label{eq:evolution}
	\ket{\tau_x} = U^T O_x U^{T-1} O_x\ldots U^1 O_x U^0 \ket{\xi}.
\end{equation}
The algorithm's goal is to determine $f(x)$ with error probability at most $1/2-\delta$ for a given $\delta>0$ and any $\did$. By standard quantum information theory,\footnote{More precisely, because of the deferred measurement principle and the fact that the ancilla space $\cvec{W}$ is arbitrarily large and any quantum positive-operator valued measurement (POVM) can be implemented by a unitary acting on an ancilla space initialized in a determined state and a projection-valued measurement, see \cite{nielsen2002quantum}.} we may assume w.l.o.g. that it does so by a single projection-valued measurement (PVM) with outcomes in $C$ performed after applying the last unitary. Therefore, if the algorithm is successful, there exists a collection of projectors $\left( P_c\right)_{c\in C}$ with each $P_c$ acting on $\cvec{Q\times W}$, $P_c P_{c'}=0$ for any $c\neq c'$, and
\begin{equation}\label{eq:function_measurements}
	\norm{P_{f(x)} \ket{\tau_x}}^2 \geq \frac{1}{2}+\delta ~~~~~ \forall \did.
\end{equation}
Assuming a collection of $\ket{\tau_x}$ and $P_c$, we define a vector $\ket{\ps}\in\cvec{D}$ whose entries are the square roots of success probabilities, i.e.
\begin{equation}
	\braket{x}{\ps} := \norm{P_{f(x)} \ket{\tau_x}}.
\end{equation}
\section{The new lemma: From perfect initial to perfect final states}\label{sec:the-new-lemma-from-perfect-initial-to-perfect-final-states}
At this point, we may already formalize the novelty in this article's approach by the following lemma:
\begin{lemma}\label{lem:initial_to_final}
Given $f\colon D\to C$, consider the problem of evaluating $f$ with standard oracles as described in Section \ref{sec:the-problem-function-evaluation}. Suppose there is a quantum query algorithm solving this problem in $T$ queries with success bias $\delta>0$, i.e. a query algorithm that transforms $\ket{\xi}$ to $\seqd{\ket{\tau_x}}$ according to Equation \ref{eq:evolution} such that projectors $\left( P_c\right)_{c\in C}$ fulfilling Equation \ref{eq:function_measurements} for that $\delta$ exist. Then there exist collections of normalized modified initial and final states $\seqd{\ket{\xi'_x}}$, $\seqd{\ket{\tau'_x}}$ such that
\begin{enumerate}
	\item \label{enum:lemma_evolution} Applying the algorithm to the $\seqd{\ket{\xi'_x}}$ results in final states $\seqd{\ket{\tau'_x}}$, i.e.
	\begin{equation}
		\ket{\tau'_x} =  U^T O_x U^{T-1} O_x\ldots U^1 O_x U^0 \ket{\xi'_x} ~~~~~ \forall\did,
	\end{equation}
\item \label{enum:lemma_perfectfinal} the states $\seqd{\ket{\tau'_x}}$ corresponding to different function values are orthogonal, i.e. can be perfectly distinguished:
\begin{equation}
	\braket{\tau'_x}{\tau'_{x'}} = 0~~~~~ \forall x,x'\in D\colon f(x)\neq f(x'),
\end{equation}
\item \label{enum:lemma_goodinitial} the modified initial states $\seqd{\ket{\xi'_x}}$ have high overlap with the $\ket{\xi}$, specifically,
\begin{equation}\label{eq:high_overlap}
	\braket{\xi}{\xi'_x} = \norm{P_{f(x)} \ket{\tau_x}} = \braket{x}{\sqrt{p_\mathrm{succ}}}\geq \sqrt{1/2+\delta} ~~~~~ \forall\did.
\end{equation}
\end{enumerate}
\end{lemma}
\begin{proof}
	For any $\did$, introduce the shorthand
	\begin{equation*}
		A_x := U^T O_x U^{T-1} O_x\ldots U^1 O_x U^0
	\end{equation*}
	for the (unitary) total evolution operator and define
	\begin{align}
		\ket{\tau'_x} &:= \frac{P_{f(x)} \ket{\tau_x}}{\norm{P_{f(x)}\ket{\tau_x}}},\\
		\ket{\xi'_x} &:= A_x^{\dagger} \ket{\tau'_x}.
	\end{align}
	In words, we defined $\seqd{\ket{\tau'_x}}$ by projecting the $\seqd{\ket{\tau_x}}$ onto the "correct" subspaces and renormalizing, and $\seqd{\ket{\xi'_x}}$ by uncomputing. Then properties \ref{enum:lemma_evolution} and \ref{enum:lemma_perfectfinal} are clear from these definitions and the fact that $A_x$ is unitary, i.e. that its Hermitian conjugate is also its inverse. To show property \ref{enum:lemma_goodinitial}, compute using Equations \ref{eq:evolution} and \ref{eq:function_measurements} that
	\begin{align}
		\braket{\xi}{\xi'_x} &= \bra{\xi} A_x^\dagger \ket{\tau'_x} = (A_x\ket{\xi})^\dagger \ket{\tau'_x}=\braket{\tau_x}{\tau'_x} \\
		&= \frac{\bra{\tau_x} P_{f(x)} \ket{\tau_x}}{\norm{P_{f(x)}\ket{\tau_x}}} = \norm{P_{f(x)}\ket{\tau_x}} \geq \sqrt{1/2+\delta}.\label{eq:computing_xixipoverlap}
	\end{align}
\end{proof}
\section{A matrix norm inequality}\label{sec:a-matrix-norm-inequality}
\begin{table}[h!]
	\centering
\begin{tabular}{| m{2cm} | m{5cm} |}
	\hline
$\herm{D}$ & Set of Hermitian complex $D\times D$ matrices \\
\hline
$\psd{D}$ & Set of positive semidefinite complex $D\times D$ matrices \\
\hline
$H[x,x']$ & Entry at $x$th row, $x'$th column of matrix $H$\\
\hline
$G_\xi\in \psd{D}$ & Gram matrix of state collection $\seqd{\ket{\xi_x}}$, i.e. $G_\xi[x,x']:=\braket{\xi_x}{\xi_{x'}}$; always in $\psd{D}$ by \cite[Chapter 1.1, Item (vi)]{bhatia2009positive} \\
\hline
$H_1 \circ H_2$ & Hadamard (entrywise) product of matrices $H_1$, $H_2$  \\
\hline
$\norm{H}$ & Matrix norm (largest singular value) of $H\in \mathbb{C}^{D\times D}$\\	
\hline
\end{tabular}
\caption{Generic linear algebra notation, relevant starting from Section \ref{sec:a-matrix-norm-inequality}. See also Table \ref{table:notation_funceval}.}\label{table:notation_linalg}
\end{table}

Using the notation of Table \ref{table:notation_linalg}, the task of this section\footnote{I am not satisfied with this proof and hope that there is a simpler way to --- possibly by referring to a different literature result.} is to prove the following lemma, which we'll be able to plug into the proof of our main Theorem \ref{thm:adversary_function}.
\begin{lemma}\label{lem:matrix_norm_inequality}
	Let $G_{\xi'}$ be the Gram matrix of a state collection $\seqd{\ket{\xi'_x}}$ that was obtained from Lemma \ref{lem:initial_to_final} with some $\delta>0$, i.e. fulfills Property \ref{enum:lemma_goodinitial} of that lemma with some normalized $\ket{\xi}\in \cvec{D}$. Then for any $\Gamma \in \herm{D}$,
	\begin{equation}
		2\delta \norm{\Gamma} \leq \norm{\Gamma \circ G_{\xi'}}.
	\end{equation}
\end{lemma}
We'll first reduce this lemma to other lemmata.
\begin{proof}[Outline]
By a triangle inequality,
\begin{align}
	\norm{\Gamma\circ \kbp} &= \norm{\Gamma \circ \left(G_{\xi'} - \left(G_{\xi'} - \kbp\right) \right)}\\
	&\leq \norm{\Gamma \circ G_{\xi'}} + \norm{\Gamma\circ \left(G_{\xi'} - \kbp\right)}.
\end{align}
This implies that
\begin{equation}
	\norm{\Gamma\circ \kbp} - \norm{\Gamma\circ \left(G_{\xi'} - \kbp\right)} \leq \norm{\Gamma \circ G_{\xi'}}.
\end{equation}
Lemmata \ref{lem:multp} and \ref{lem:multgdelta} bound the two summands on the left-hand side. Together, they imply that $\norm{\Gamma \circ G_{\xi'}}$ is lower-bounded by
\begin{equation}
	\left(2 \min_{\did} \left\vert{\braket{x}{\ps}}\right\vert^2 - 1\right)\norm{\Gamma},
\end{equation}
which in turn is lower-bounded by $2\delta \norm{\Gamma}$.
\end{proof}

It remains to state and prove Lemmata \ref{lem:multp} and \ref{lem:multgdelta}. Our approach crucially relies on a theorem by Schur, which we state without proof:
\begin{thm}[\cite{bhatia2009positive}, Theorem 1.4.1]\label{thm:psd_schurnorm}
	Let $\Gamma\in \herm{D}$ be Hermitian and $S \in \psd{D}$ positive semidefinite. Then
	\begin{equation}
		\norm{\Gamma \circ S} \leq \max_{\did} S[x,x] \norm{\Gamma}.
	\end{equation}
\end{thm}

\begin{lemma}\label{lem:multp}
	Let $\ket{\ps}\in\cvec{D}$ with all entries nonzero, and $\Gamma \in \herm{D}$. Then
	\begin{equation}
		 \min_{\did} \left\vert{\braket{x}{\ps}}\right\vert^2 \norm{\Gamma}\leq \norm{\Gamma \circ \kbp}.
	\end{equation}
\end{lemma}
\begin{proof}
	Set $\Gamma' := \Gamma \circ \kbp[\ps]$. Let $\ket{\ps^{-1}}\in \cvec{D}$ be the vector whose entries are inverse to those of $\ket{\ps}$. Then the claim is equivalent to
	\begin{align}
		&\norm{\Gamma' \circ \kbp[\ps^{-1}]} \leq \frac{\norm{\Gamma'}}{\min_{\did} \left\vert{\braket{x}{\ps}}\right\vert^2 } \\
		=& \max_{\did} \left(\kbp[\ps^{-1}]\right)\left[x,x\right] \norm{\Gamma'}.
	\end{align}
As $\kbp[\ps^{-1}]$ is positive semidefinite, this is an instance of Theorem \ref{thm:psd_schurnorm}.
\end{proof}

\begin{lemma}\label{lem:multgdelta}
	Let $\Gamma\in \herm{D}$ and $G_{\xi'}$ be the Gram matrix of a state collection as in Lemma \ref{lem:initial_to_final}. Then
	\begin{equation}
		-\left(1 - \min_{\did} \left\vert{\braket{x}{\ps}}\right\vert^2\right) \norm{\Gamma} \leq -\norm{\Gamma \circ \left( G_{\xi'} - \ket{\ps}\bra{\ps} \right)}.
	\end{equation}
\end{lemma}
\begin{proof}
	Set $\Delta G :=  G_{\xi'} - \kbp $. We deduce the claim from Theorem \ref{thm:psd_schurnorm} and the following two facts:
	\begin{enumerate}
		\item $\Delta G$ is positive semidefinite. To see this, recall that $\braket{\xi'_x}{\xi} = \braket{x}{\ps}$ for all $\did$ and the entries of $\Delta G$ fulfill
		\begin{align}
	\Delta G[x, x'] &= \braket{\xi'_x}{\xi'_{x'}} - \braket{\xi'_x}{\xi}\braket{\xi}{\xi'_{x'}}\\
	& = \braket{\xi'_x}{I - P_\xi \mid \xi'_{x'}}.
	\end{align}
	Therefore, $\Delta G$ is the Gram matrix of the vectors $\seqd{\left(I - P_\xi\right)\ket{\xi'_x}}$, so it is positive semidefinite.
	\item The diagonal entries of $\Delta G$ are at most
	\begin{equation}
		1 - \min_{\did} \left\vert{\braket{x}{\ps}^2}\right\vert = \max_{\did} \left(1- \left\vert{\braket{x}{\ps}^2}\right\vert\right).
	\end{equation}
This follows from the fact that the diagonal entries of $G_{\xi'}$ are $1$, as the vectors $\ket{\xi'_{x}}$ are normalized.
	\end{enumerate}
\end{proof}
\section{The adversary bound for state conversion}\label{sec:the-adversary-bound-for-state-conversion}
The adversary bound for state conversion \cite{lee2011quantum} gives a lower bound on the number of queries necessary to accomplish a so-called \emph{state conversion} task. The version we need follows from Lemma 4.7 in that work, but note that it uses slightly different conventions: It defines a "filtered $\gamma_2$-norm", which is twice of what we will call the adversary bound. For another reference, see \cite[Theorem 10 and Remark 9]{belovs2015variations}.
\begin{thm}[{follows from \cite[Lemma 4.7]{lee2011quantum}}]\label{thm:adversary_state}
 For $0\leq i < n$, define $\Delta_i \in \herm{D}$ by $\Delta_i[x,x'] := 0$ if $x_i = x'_i$ and $2$ otherwise.\footnote{In words, denoting whether one is able to distinguish $x$ from $x'$ by querying position $i$ of the query register.} Then any quantum query algorithm that converts an $x\in D$-dependent initial state collection $\seqd{\ket{\xi_x}}$ to a final state collection $\seqd{\ket{\tau_x}}$ must take at least
\begin{equation}\label{eq:def_adv}
	\adv{G_\xi - G_\tau} := \max_{\Gamma\in\herm{D}}\frac{\norm{\Gamma \circ \left(G_\xi - G_\tau\right)}}{\max_{0\leq i < n} \norm{\Delta_i \circ \Gamma}}
\end{equation}
queries, where $G_\xi$, $G_\tau\in\herm{D}$ are the Gram matrices of $\seqd{\ket{\xi_x}}$ and $\seqd{\ket{\tau_x}}$ as in Table \ref{table:notation_linalg}.
\end{thm}
\section{Putting it together}\label{sec:putting-it-together}
Using Lemmata \ref{lem:initial_to_final} and \ref{lem:matrix_norm_inequality} and the adversary bound for state conversion, Theorem \ref{thm:adversary_state}, we now prove the central claim of our paper. We define an adversary bound for function evaluation in the statement of the next theorem; however, note that it is smaller by a factor of $2$ than the definition used in \cite{childs2017lecture}.
\begin{thm}\label{thm:adversary_function}
	Consider the problem of evaluating $f\colon D\to C$ with standard oracles and success probability at least $1/2+\delta$ for all $\did$ as in Section \ref{sec:the-problem-function-evaluation}.
	Let $\mathbb{F}_f\subset \herm{D}$ be the space of Hermitian matrices whose $(x,x')$ entries are $0$ whenever $f(x) = f(x')$. With $\Delta_i$ as in Theorem \ref{thm:adversary_state},
	\begin{equation}
		2 \delta \adv{f} := 2 \delta 	\max_{\Gamma\in\mathbb{F}_f}\frac{\norm{\Gamma}}{\max_{0\leq i < n} \norm{\Delta_i \circ \Gamma}}
	\end{equation}
	lower-bounds the number of queries that an algorithm solving this problem must take.
\end{thm}
\begin{proof}
	Consider $\Gamma\in \mathbb{F}_f$ with
	\begin{equation}\label{eq:allsmallerone}
		\max_{0\leq i < n} \norm{\Delta_i \circ \Gamma} \leq 1,
	\end{equation}
	and assume that the problem can be solved in $T$ queries. Use Lemma \ref{lem:initial_to_final} to find a $T$-query quantum algorithm converting $\seqd{\ket{\xi'_x}}$ to $\seqd{\ket{\tau'_x}}$ as in that lemma. By Theorem \ref{thm:adversary_state}, this implies that
	\begin{equation}\label{eq:inequality_first}
		\norm{\Gamma \circ \left(G_{\xi'} - G_{\tau'}\right)} \leq T.
	\end{equation}
	For all $x,x'\in D$ with $f(x) \neq f(x')$, $G_{\tau'}[x,x'] =0$ (Property \ref{enum:lemma_perfectfinal} of Lemma \ref{lem:initial_to_final}). Furthermore, for all $x,x'\in D$ with $f(x) = f(x')$, $\Gamma[x,x'] = 0$ by assumption. This implies that $\Gamma \circ G_{\tau'} = 0$, and Equation \ref{eq:inequality_first} implies
	\begin{equation}
		\norm{\Gamma\circ G_{\xi'}} \leq T.
	\end{equation}
	Using Lemma \ref{lem:matrix_norm_inequality}, we conclude
	\begin{equation}
		2 \delta \norm{\Gamma} \leq \norm{\Gamma \circ G_{\xi'}} \leq T.
	\end{equation}
	As any $\Gamma \in \mathbb{F}_f$ can be obtained by rescaling some $\Gamma \in \mathbb{F}_f$ fulfilling Equation \ref{eq:allsmallerone}, the theorem follows.
\end{proof}
\section{Discussion and outlook}\label{sec:discussion}
The previous approach to prove the analogue to Theorem \ref{thm:adversary_function} (with the old prefactor) essentially involved bounding the quantity
\begin{equation}
	\norm{\Gamma \circ \left(G_\xi - G_\tau\right)},\label{eq:oldadv}
\end{equation}
in other words, the state conversion adversary bound corresponding to the algorithm applied to the true initial state (rather than the modified ones we obtained from Lemma \ref{lem:initial_to_final}). In the standard derivation, this leads to the old prefactor of $2\delta -\sqrt{1-4\delta^2}$, or $1-\sqrt{1-4\delta^2}$ in the case of Boolean functions \cite{childs2017lecture, negative_weights}. To my understanding (though I didn't prove this), Equation \ref{eq:oldadv} does not obey a bound as it is the result of this article.

The adversary bound is described in terms of the Gram matrices of initial and final states, and the proof may also be interpreted in terms of these (as discussed e.g. in \cite{yolcu2022adversary}). But similarly to the universal query algorithm of \cite{belovs2023one, yolcu2022adversary}, it turned out to be fruitful to think about both the "Gram matrix picture" and the "states-and-unitaries picture", using the latter to change the Gram matrices under consideration to more convenient ones. In our case, the convenience came from the perfect distinguishability of the modified final states, which corresponds to Gram matrices whose entries are exactly $0$. The interplay between these pictures --- and the way that using both in a proof can be helpful --- continues to mystify me.

As remarked, the proof of Lemma \ref{lem:matrix_norm_inequality} seems overly complicated to me, and I am not satisfied with it. Though I see some other approaches to it, they are not simpler at this point. I hope that one may prove Lemma \ref{lem:matrix_norm_inequality} in a more elegant way, with or without a literature result. That may allow an argument deriving the function evaluation adversary bound from scratch that is shorter than the approach given e.g. in \cite{childs2017lecture}.

In the history of the adversary bound, it has been a recurring theme that progress came from improved understanding of the output condition, i.e. the condition on state collections that allow determining a function value with nonzero error --- for example, the negative-weights adversary method of \cite{negative_weights} was possible due to a refined output condition. This paper fits into that theme. Belovs \cite{belovs2015variations} introduced "purifiers", which are another approach to the finite-error-probability problem, but doesn't give a constant factor. It would be interesting to check how the constant factor compares to the textbook one, and if the purifier framework can be put combined with this result.
\section{Acknowledgements}
I thank Aleksandrs Belovs for helpful discussions on this subject.
\bibliographystyle{plain}
\bibliography{adversary_better_prefactor}
\end{document}